\documentclass[12pt]{article}
\usepackage{amsmath}
\usepackage{amsfonts}
\usepackage{amssymb}
\usepackage{amsthm}
\usepackage{fullpage}
\usepackage{graphicx}
\def\be{\begin{equation}}
\def\ee{\end{equation}}
\def\bea{\begin{eqnarray}}
\def\eea{\end{eqnarray}}
\def\bma{\begin{mathletters}}
\def\ema{\end{mathletters}}

\def\0{\overline{0}}
\def\tr{\mbox{tr}}

\def\q0{\underline{0}}

\def\H{{\cal H}}

\def\id{{\mathbb I}}

\def\H{{\cal H}}

\def\tr{\mbox{tr}}
\def\one{\leavevmode\hbox{\small1\normalsize\kern-.33em1}}

\def\bra#1{\langle#1|} \def\ket#1{|#1\rangle}

\def\proj#1{\ket{#1}\!\bra{#1}}

\newtheorem{theo}{Theorem}

\newtheorem{remark}{Remark}
\newtheorem{defin}[theo]{Definition}

\newtheorem{lemma}[theo]{Lemma}

\def\id{{\mathbb I}}
\def\tr{\mbox{tr}}

\begin{document}
\title{A physical approach to Tsirelson's problem}
\author{M. Navascu\'es, T. Cooney, D. P\'erez-Garc\'ia and N. Villanueva\\Facultad de Matem\'aticas, Universidad Complutense de Madrid}
\date{}
\maketitle

\begin{abstract}
Tsirelson's problem deals with how to model separate measurements in quantum mechanics. In addition to its theoretical importance, the resolution of Tsirelson's problem could have great consequences for device independent quantum key distribution and certified randomness. Unfortunately, understanding present literature on the subject requires a heavy mathematical background. In this paper, we introduce quansality, a new theoretical concept that allows to reinterpret Tsirelson's problem from a foundational point of view. Using quansality as a guide, we recover all known results on Tsirelson's problem in a clear and intuitive way.
\end{abstract}

\section{Introduction}

Standard (i.e., non-relativistic) quantum mechanics describes bipartite systems by associating them to a tensor product of Hilbert spaces $\H=\H_A\otimes\H_B$. Measurements and operations corresponding to party $A$ ($B$) are then described by linear operators on $\H$ that act trivially over $\H_B$ ($\H_A$). Things are different, though, in Algebraic Quantum Field Theory (AQFT), where the Haag-Kastler axioms only require that operators corresponding to measurements performed in space-like separated bounded regions of Minkowski space-time commute \cite{AQFT}.

Tsirelson's problem consists in finding out if the bipartite correlations generated one way or the other are essentially the same, i.e., if any set of AQFT-like correlations admits a non-relativistic approximation. The problem of bounding the set of tensor bipartite correlations arises naturally in device-independent quantum key distribution \cite{dev_QKD} and certified randomness generation \cite{dev_rand}. Since current theoretical tools can only bound the set of AQFT-like correlations \cite{hier}, we may be unnecessarily limiting the range of all such protocols.

Tsirelson himself proved that commutative and tensor sets of correlations coincide if the underlying Hilbert space $\H$ where measurement operators act is finite dimensional (see \cite{werner}). Sholtz and Werner then showed that this result still holds under the weaker assumption that Alice's or Bob's operator algebra is nuclear \cite{werner}. Such is the case if, for instance, Alice is restricted to perform two different dichotomic measurements. More recently, Tsirelson's problem was related to Connes' embedding conjecture in $C^*$-algebras \cite{connes}. The above results rely on very heavy mathematical machinery; their proofs (and sometimes even statements) thus appear incomprehensible to the average physicist.

Tsirelson's problem has recently gained popularity in the Quantum Information community, not only due to the practical consequences of its resolution, but also for the aforementioned connections with sophisticated areas of mathematics. Looking at it from the outside, though, that Tsirelson's problem is related to a conjecture which hundreds of mathematicians have attempted to solve over the course of 35 years is not good news at all. It implies that a physicist aiming at solving it should study advanced mathematics for several years in order to find himself as lost as the very brilliant mathematical minds who are currently trying to solve Connes' embedding conjecture.

In this paper we will introduce a new foundational concept called \emph{quansality}, related to the possibility of finding a local quantum description for experiments in a bounded region of space-time. We will show that, despite its intuitive nature, quansality can be proven mathematically equivalent to a tensor product description of bipartite quantum correlations. Therefore, if there existed an AQFT-type scenario giving rise to bipartite correlations impossible to reproduce in non-relativistic quantum systems, then any quantum model attempting to explain the measurement statistics of one the parts could be experimentally falsified.

This intuition will allow us to recover all known results on Tsirelson's theorem, with physical insight and basic linear algebra as our only tools. With this, we thus hope to make Tsirelson's problem accessible to any physicist with a reasonable mathematical knowledge.

The paper is organized as follows: in Section \ref{notation} we will introduce the notation to be used along the article. Then we will define and motivate the concept of quansality. In Section \ref{tsirelson_prob} we will present Tsirelson's problem and derive Tsirelson's finite dimensional result by explicitly constructing a local quantum model for one of the party's observations. We will show how to extend some results to the infinite dimensional case in Section \ref{infinite}. Finally, we will present our conclusions.

\section{Mathematical notation}
\label{notation}

Along this article, given a separable Hilbert space $\H$, we will regard quantum states as positive elements of $S_1=\{A\in B(\H): \|A\|_1<\infty\}$, that we will suppose normalized unless otherwise specified. Note that $S_1\subset S_2=\{A\in B(\H): \tr(AA^\dagger)<\infty\}$. Both $S_1$ and $S_2$ will play an important role in the discussion of Tsirelson's theorem.

Being the trace norm the natural norm of the set $S_1$, we will say that a sequence of quantum states $(\rho_N)$ converges to $\rho$ in $S_1$ if $\lim_{N\to\infty}\|\rho_N-\rho\|_1=0$. Analogously, we will say that $(\rho_N)$ converges to $\rho$ in $S_2$ if $\lim_{N\to\infty}\tr\{(\rho_N-\rho)(\rho_N-\rho)^\dagger\}=0$. Notice that convergence in $S_1$ implies convergence in $S_2$; the contrary is generally false.

\section{Quansality}
\label{quansality}

Assume that we are in a scenario where two separated parties, Alice and Bob, are each able to tune interactions $x$ and $y$ in their respective labs (see Figure \ref{Alice_Bob}). Such interactions (measurements) will in turn produce two pieces of experimental data $a$ and $b$ (the outcomes). We will make the realistic assumption that the sets $A$ and $B$ of possible interactions which Alice and Bob can produce, as well as the set of possible experimental results they can observe, are finite. That way, if Alice and Bob compare their measurement results after many repetitions of the experiment, they will be able to estimate the set of probability distributions $\{P(a,b|x,y):x\in A, y\in B\}$. Following Tsirelson's terminology, we will call such a set \emph{behavior}, and we will denote it by $P(a,b|x,y)$. This notation will not lead to confusion, since along the text it will always be clear whether by $P(a,b|x,y)$ we refer to a probability or a behavior.

\begin{figure}
  \centering
  \includegraphics[width=12 cm]{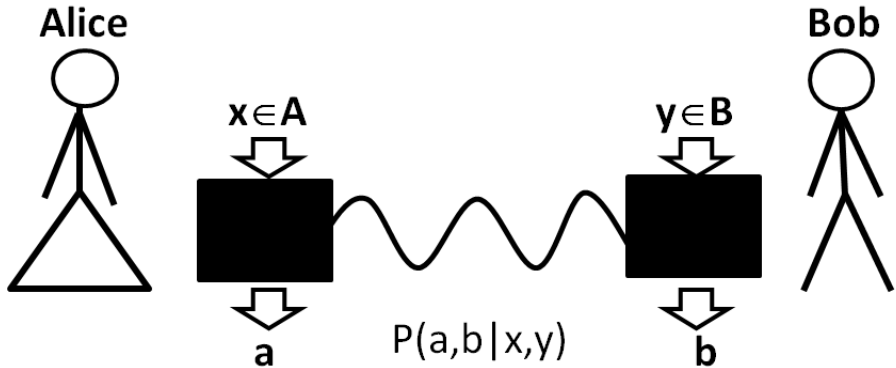}
  \caption{\textbf{Sketch of a bipartite experiment of non-locality.}}
  \label{Alice_Bob}
\end{figure}

Now, suppose that Alice performs a measurement $x$ with outcome $a$, and denote by $P^x_a(b|y)$ Bob's resulting single behavior. Then the fact that Bob's statistics can be described locally is equivalent to the condition

\be
\sum_{a}P(a|x)P^x_a(b|y)=P(b|y),
\label{causality}
\ee

\noindent where $P(b|y)$ is a probability distribution independent of $x$. This condition ensures that Bob's statistics are not affected by Alice's measurement choice $x$, and so Bob can study his system without introducing Alice as a variable. Taking into account that $P(a|x)p^x_a(b|y)=P(a,b|x,y)$, we have that the last condition is equal to

\be
\sum_{a}P(a,b|x,y)=P(b|y),
\ee

\noindent the no-signaling condition \cite{popescu}.
\vspace{10pt}

Suppose now that we seek a description of Bob's part that is not only independent of Alice, but also \emph{compatible with Quantum Mechanics} (QM). How should we modify condition (\ref{causality}) to account for this extra condition? The answer requires a detailed discussion.

To build a quantum model for his system, Bob must find a set of projector operators $\{F^y_b\}\subset B(\H_B)$, with $\sum_b F^y_b=\id_{\H_B}$, and a quantum state $\sigma\in S_1(\H_B)$ such that $P(b|y)=\tr(\sigma F^y_b)$. Suppose he does. 

Now, imagine that Alice switches interaction $x$ and reveals Bob both $x$ and her measurement results in all realizations of the experiment. If Bob's model is truly independent of Alice, her distant actions should not alter the previous picture. Consequently, if his initial model was correct, then Bob should be able to find states $\{\bar{\sigma}^x_a\}\subset S_1(\H_B)$ such that $P(b|y,a,x)=\tr(\bar{\sigma}^x_aF^y_b)$. 

If for some choice of $\{\bar{\sigma}^x_a\}$ we have that $\sum_{a}P(a|x)\bar{\sigma}^x_a=\sigma$, then Bob's model prior to Alice's communications can be considered correct; Alice is just completing such a model with supplementary information $(x,a)$ about Bob's physical states. If, however, $\sum_{a}P(a|x)\bar{\sigma}^x_a=\sigma$ cannot be enforced for any choice of $\{\bar{\sigma}^x_a\}$ compatible with the statistics $P(b|y,a,x)$, then Bob's model was flawed from the very beginning, and Alice has helped him to falsify it.

This way, we arrive at the notion of \emph{quansality}.

\begin{defin}\textbf{Quansality}\\
Let $P(a,b|x,y)$ be a bipartite behavior. We say that $P(a,b|x,y)$ is quansal (QUANtum cauSAL) if there exist a Hilbert space $\H_B$, projector operators $\{F^y_b:\sum_b F^y_b=\id_{\H_B}\}$ and a set of subnormalized quantum states $\{\sigma^x_a\}\subset S_1(\H_B)$ such that $P(a,b|x,y)=\tr(F^y_b\sigma^x_a)$ and

\be
\sum_{a}\sigma^x_a=\sigma\not=\sigma(x).
\label{q_causality}
\ee

\end{defin}

\noindent From the above discussion, it follows that quansality is a necessary and sufficient condition for Bob to find an independent quantum description of his physical system. 

Note that quansality implies the no-signaling condition (from Alice to Bob). Condition (\ref{q_causality}), though, is not necessary to prevent signaling from Alice, since it could well be that, although $\sum_a\sigma^x_a\not=\sum_a\sigma^{x'}_a$ for $x\not=x'$, both states exhibit the same statistics when measured with $\{F^y_b\}$. In the next section, we will explain how quansality relates to Tsirelson's problem. We will also prove that quansality is a symmetric property, i.e., if Bob's subsystem admits a quantum model, so does Alice's.

\section{Tsirelson's problem}
\label{tsirelson_prob}

\begin{defin}
Let $P(a,b|x,y)$ be a bipartite behavior. $P(a,b|x,y)$ is non-relativistic (or admits a tensor representation) if there exist a pair of Hilbert spaces $\H_A,\H_B$, a state $\rho_{AB}\in S_1(\H_A\otimes\H_B)$ and two sets of projectors $\{E^a_x\in B(\H_A)\}$, $\{F^b_y\in B(\H_B)\}$ such that

\begin{enumerate}
\item $\sum_a E^x_a=\id_A,\sum_b F^y_b=\id_B$, for all $x,y$.
\item $\tr(\rho_{AB}E^x_a\otimes F^y_b)=P(a,b|x,y)$.
\end{enumerate}

The set of all behaviors of this form will be denoted by $Q$.

\label{tensorial}
\end{defin}

\begin{defin}
Let $P(a,b|x,y)$ be a bipartite behavior. We will say that $P(a,b|x,y)$ is relativistic (or admits a field representation) if there exists a separable Hilbert space $\H$, a quantum state $\rho\in S_1(\H)$ and a set of projectors $\{E^x_a,F^y_b:x\in A,y\in B\}$, such that

\begin{enumerate}
\item $\sum_a E^x_a=\sum_b F^y_b=\id$, for all $x,y$.
\item $[E^x_a,F^y_b]=0$, for all $a,b,x,y$.
\item $\tr(\rho E^x_aF^y_b)=P(a,b|x,y)$.
\end{enumerate}

We will call $Q'$ the set of all such behaviors. In \cite{hier} it was proven that $Q'$ is a closed set.

\label{commutative}
\end{defin}

\begin{remark}

It can be proven (see \cite{connes}) that in the definitions of $Q$ and $Q'$ we can relax the condition that measurement operators are projectors and demand instead that such operators $E^x_a,F^y_b$ are positive semidefinite, i.e., elements of a Positive Operator Valued Measure (POVM).

\end{remark}

Since $[E^x_a\otimes \id_B,\id_A\otimes F^y_b]=0$, we have that $\bar{Q}\subset Q'$, i.e., any behavior approximable by non-relativistic behaviors admits a field representation. Tsirelson's problem consists of deciding whether such an inclusion is strict, that is, whether $\bar{Q}\not=Q'$ \cite{werner}. In addition to its relevance in Algebraic Quantum Field Theory (AQFT), Tsirelson's problem acquires a foundational nature due to the following lemma (inspired by the constructions presented in \cite{gleason,gleason2}).

\begin{lemma}
$P(a,b|x,y)$ is non-relativistic iff it is quansal.
\label{equivalence}
\end{lemma}

\begin{proof}
The right implication is immediate: if $P(a,b|x,y)$ satisfies the conditions of Definition \ref{tensorial}, take $\sigma^x_a\equiv\tr_A(\rho_{AB}E^x_a\otimes\id_B)$ and note that $\sum_a\sigma^x_a=\rho_B$ and that $\tr(\sigma^x_aF^y_b)=\tr(\rho_{AB}E^x_a\otimes F^y_b)=P(a,b|x,y)$.

For the left implication, define

\be
E^x_a\equiv [\sigma^{-1/2}\sigma^x_a\sigma^{-1/2}]^T,
\ee

\noindent where $T$ denotes the transpose with respect to an orthonormal basis $\{\ket{j}\}_j$ for $\H$. Clearly, $E^x_a\geq 0$, and

\be
\sum_aE^x_a=\sum_a[\sigma^{-1/2}\sigma^x_a\sigma^{-1/2}]^T=\id,
\ee

\noindent and so Alice's operators are complete POVM elements. Define now the quantum state $\rho_{AB}\equiv\ket{\phi_\sigma}\bra{\phi_\sigma}$ with

\be
\ket{\phi_\sigma}= \sum_{j=1}^{\mbox{dim}(\H)}\ket{j}\otimes\sigma^{1/2}\ket{j}.
\ee

\noindent Here, we allow for the possibility that $\mbox{dim}(\H)=\infty$ (we are working under the assumption that $\H$ is separable). 

Obviously, $\rho_{AB}\geq 0$. Also, $\tr(\rho_{AB})=\tr_B(\sigma^{1/2}\id\sigma^{1/2})=\tr(\sigma)=1$. The operator $\rho_{AB}$ thus corresponds to a normalized quantum state. Finally, we have that

\begin{eqnarray}
&&\tr(\rho_{AB}E^x_a\otimes F^y_b)=\tr_B(\sigma^{1/2}(E^x_a)^T\sigma^{1/2}F^y_b)=\nonumber\\
&&=\tr(\sigma^x_aF^y_b)=P(a,b|x,y).
\end{eqnarray}

\end{proof}

Therefore, if $\bar{Q}\not= Q'$, then there exist multipartite scenarios in Algebraic Quantum Field Theory where each part cannot describe its own local quantum system in a way consistent with future communications from other parties. Such a situation would be certainly undesirable: if that were the case, then any quantum theory attempting to explain the results of a particular experiment would have to include the rest of the universe in its formulation! 

In the following sections we will show how to recover all existing results on Tsirelson's theorem by building upon this intuition. That is, given a relativistic bipartite distribution, we will find a local quantum model consistent with Bob's measurement statistics, thereby proving that such a distribution also admits a tensor representation.

\section{Finite dimensions and the Heat Vision effect}
In the following, we will prove that definitions \ref{tensorial} and \ref{commutative} are equivalent when we demand the Hilbert spaces $\H,\H_A,\H_B$ to be finite dimensional.

\begin{theo}{\textbf{Tsirelson's Theorem}\\}
Let $P(a,b|x,y)$ admit a finite dimensional field representation. Then, $P(a,b|x,y)$ is non-relativistic.
\label{finite}
\end{theo}

\begin{proof}

From lemma \ref{equivalence}, it is enough prove that $P(a,b|x,y)$ is quansal, i.e., that there exist operators $\{F^y_b\}$ and quantum states $\{\sigma^x_a\}$ giving rise to $P(a,b|x,y)$ and satisfying (\ref{q_causality}).

To grasp some intuition about the problem, picture the following scenario: Alice performs a measurement and then gives Bob control over the whole space $\H$. Consequently, $\H_B\to\H$, $\sigma^x_a\to\rho^x_a\equiv E^x_a\rho E^x_a$ and Bob's measurement operators are $\{F^y_b\}$. Let us see how this rough approach performs. Indeed, we have that $P(a,b|x,y)=\tr(\rho^x_aE^y_b)$. Nevertheless, the states $\sum_a \rho^x_a\equiv\rho^x$ will be generally $x$-dependent. The problem of this latter strategy is thus that Alice's initial measurement leaves a `mark' $x$ on the state, and this mark is to be erased if we want to recover Tsirelson's theorem. 

Let us then complicate slightly our previous scheme, see Figure \ref{fury}: as before, Alice is going to measure, store the pair of values $(x,a)$ and give Bob total control of $\H$. However, this time she's going to introduce noise to her system before leaving the house. Note that the Kraus operators of any Completely Positive (CP) map Alice applies to the states $\rho^x_a$ will commute with Bob's measurement operators. Thus the statistics of any perturbed state $\tilde{\rho}^x_a$ will be the same, i.e., $\tr(\tilde{\rho}^x_aF^y_b)=\tr(\rho^x_aF^y_b)=P(a,b|x,y)$. On the other hand, it could be that $\tilde{\rho}^x$ is closer to $\tilde{\rho}^{x'}$ than $\rho^x$ to $\rho^{x'}$.

\begin{figure}
  \centering
  \includegraphics[width=16 cm]{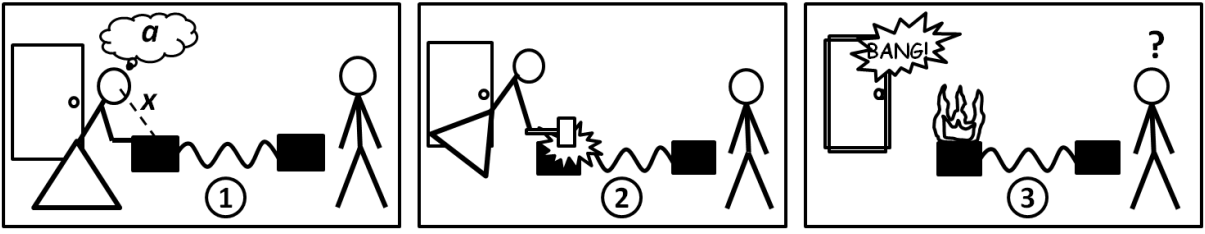}
  \caption{\textbf{How to build a quansal representation for $P(a,b|x,y)\in Q'$.} 1) Alice interacts with her system in some way $x$ and stores the outcome $a$. 2) She adds noise to her part of the lab in order to erase the information of her measurement choice. Since she is restricted to act with her algebra of observables, her operations do not affect Bob's statistics. 3) She gives Bob full control of the bipartite system.}
  \label{fury}
\end{figure}

What could Alice do? Well, we just know that she is capable of implementing a number of measurements. Suppose, then, that after performing a measurement $x$ and storing the outcome $a$, Alice measures randomly her system $N$ times in the hope of erasing the information about $x$. 

Define the completely positive maps $\Omega_x:B(\H)\to B(\H)$

\be
\Omega_x(M)\equiv \sum_{a} E^x_a M E^x_a.
\ee

\noindent The map

\be
\Omega(M)\equiv \frac{1}{|A|}\sum_{x\in A} \Omega_x(M)
\label{eraser}
\ee

\noindent therefore represents the channel that results when Alice applies any interaction $x$ with probability $1/|A|$. We will next see that the limit $\Omega^\infty(\rho)\equiv\lim_{N\to \infty}\Omega^N(\rho)$ always exists, and, indeed, it erases the information on $x$. That is, $\Omega^\infty(\Omega_x(\rho))=\Omega^\infty(\rho)$, for all $x\in A$.

Through the invertible transformation $\ket{i}\bra{j}\to \ket{i}\ket{j}$, we can view $\rho$ as a vector in $S_2\cong\H\otimes \H$ and $\Omega$ as a super operator $\bar{\Omega}$ acting over $\H\otimes \H$. The expression of $\Omega$ thus changes to

\be
\bar{\Omega}=\frac{1}{|A|}\sum_x \bar{\Omega}_x,\mbox{ with } \bar{\Omega}_x=\sum_a E^x_a\otimes (E^x_a)^T.
\ee

Note that $\bar{\Omega}$ is a hermitian positive semidefinite operator, and so it can be diagonalized, i.e., $\bar{\Omega}=\sum_{i} \lambda_i\proj{\Psi_i}$. Moreover, since $\bar{\Omega}$ is a convex combination of projectors $\bar{\Omega}_x$, we have that $\lambda_i\in [0,1]$. It follows that $\bar{\Omega}^N=\sum_{i} \lambda^N_i\proj{\Psi_i}$, and thus

\be
\bar{\Omega}^\infty=\lim_{N\to\infty}\bar{\Omega}^N=\sum_{i|\lambda_i=1}\proj{\Psi_i}.
\ee

The existence of the map $\Omega^\infty$ in $S_2$ has been proven. Since in finite dimensions $S_1=S_2$ (i.e., the norms $\|\cdot\|_1,\|\cdot\|_2$ are equivalent), $\lim_{N\to \infty}\Omega^N(\rho)$ exists for all $\rho\in S_1$. From $\tr\{\Omega^N(\rho)\}=1$ and $(\id\otimes\Omega^N)(\rho)\geq 0$, we thus have that $\Omega^\infty$ is completely positive and trace preserving.

Suppose now that $\ket{\Psi}$ is such that $\bar{\Omega}\ket{\Psi}=\ket{\Psi}$. Then,

\be
\bra{\Psi}\bar{\Omega}\ket{\Psi}=\frac{1}{|A|}\sum_{x=1}^{|A|}\bra{\Psi}\bar{\Omega}_x\ket{\Psi}=1.
\ee

\noindent This can only be true if $\bra{\Psi}\bar{\Omega}_x\ket{\Psi}=1$ for all $x$, and so $\bar{\Omega}_x\ket{\Psi}=\ket{\Psi}$ for all $x$. This implies that $\bar{\Omega}_x\bar{\Omega}^\infty\ket{\phi}=\bar{\Omega}^\infty\ket{\phi}$ for any vector $\ket{\phi}$. We thus have that $\bar{\Omega}^\infty\bar{\Omega}_x=\bar{\Omega}^\infty$, i.e., for any $x$ and any $M\in S_2$,

\be
\Omega^\infty(\sum_a E^x_aM E^x_a)=\Omega^\infty(M).
\label{invar}
\ee

In summary, we have that the states

\be
\sigma^x_a\equiv \Omega^\infty(E^x_a\rho E^x_a)
\ee

\noindent satisfy

\be
\sum_a \sigma^x_a=\sum_a\Omega^\infty(E^x_a\rho E^x_a)=\Omega^\infty(\sum_aE^x_a\rho  E^x_a)=\Omega^\infty(\rho)\equiv \sigma,
\ee

\noindent and so generate the quansal behavior

\be
\tr(\sigma^x_aF^y_b)= \lim_{N\to\infty}(\Omega^N(E^x_a\rho E^x_a)F^y_b)=\tr(E^x_a\rho E^x_a F^y_b)=P(a,b|x,y).
\ee

\end{proof}

\begin{remark}
The previous proof can be extended to the case where $E^x_a,F^b_y$ are general POVM elements rather than projectors. The maps $\Omega_x$ should then be defined as

\be
 \Omega_x(M)\equiv\sum_a \sqrt{E^x_a}M\sqrt{E^x_a}
\ee 

Seen as operators in $S_2$, these maps are also positive semidefinite, and its norm is smaller or equal than 1, since they are unital \cite{channel}: the eigenvalues of $\Omega$ are thus in the interval $[0,1]$. From this point on, the rest of the proof is identical.
\end{remark}

Note that in order to prove Tsirelson's theorem we just invoked finite dimensionality to argue that $S_2=S_1$, and thus that the limiting state $\lim_{N\to\infty}\Omega(\rho)$ exists in $S_1$ for any initial $\rho$. In infinite dimensions, $S_2\not=S_1$ in general, so it can happen that repeated independent random measurements do not thermalize, i.e., that the map $\lim_{N\to\infty}\Omega^N$ is well defined in $S_2\cong\H\otimes\H$ but nevertheless the sequence of states $\lim_{N\to\infty}\Omega^N(\rho)$ does not converge to any quantum state (any element of $S_1$). Such a non-trivial phenomenon is called \emph{Heat Vision} \cite{heat_vision}, and, under some fair assumptions on the structure of the hamiltonian operator of the system\footnote{More concretely, that the spectrum of the hamiltonian $H$ does not contain accumulation points. This is always the case if $H$ describes a finite number of trapped non-relativistic particles subject to a potential bounded from below.}, is always accompanied by an unbounded energy increase. As shown in \cite{heat_vision}, the Heat Vision effect can be observed even in infinite dimensional systems with just two von Neumann measurements of two outcomes each.

Tsirelson's theorem can therefore be extended in the following way:

\begin{theo}{\textbf{Tsirelson's Theorem (bis)}\\}
Let $P(a,b|x,y)$ be a bipartite behavior, attainable in a separable Hilbert space by means of a commutative setup. If Alice's or Bob's measurement operators are such that random independent measurements do not produce Heat Vision, then $P(a,b|x,y)$ is non-relativistic.
\end{theo}

\section{Infinite dimensions-general results}
\label{infinite}

The previous theorem makes one wonder whether arguments based on Heat Vision are really necessary in order to guarantee that Tsirelson's theorem holds in infinite dimensions. The next result shows that such is not the case.

\begin{lemma}
Suppose that $|A|=2$, and that both measurement settings $x_1,x_2$ have associated two outcomes each. Then, any behavior $P(a,b|x,y)$ admitting a field representation can be approximated arbitrarily well by non-relativistic behaviors.
\end{lemma}

\begin{proof}
First, note that

\be
\sum_{a=1,2} E^x_a\rho E^x_a=\frac{1}{2}(U_x\rho U_x+\rho),
\ee

\noindent where $U_x= E^x_1-E^x_2$. It is immediate that $U_x=U_x^\dagger$, and $U_x^2=\id$. The operators $U_1,U_2$ are thus unitary and self-adjoint, and, of course, commute with Bob's measurement operators. We will denote by $\{\Lambda_x:x=1,2\}$ the CP maps $\Lambda_x(\rho)=U_x\rho U_x$.

Although we will follow the lines of the proof of Theorem \ref{finite}, there will be two important changes: 1) rather than using a single channel $\Omega^\infty$ to erase the information on $x$, we will consider two families of channels $\Gamma^{(N)}_1,\Gamma^{(N)}_2$, one for each measurement setting $x$. 2) The limits $\lim_{N\to\infty}\Gamma_1^{(N)}(\rho^1),\Gamma_2^{(N)}(\rho^2)$ will not exist in general. Instead, we will demand the much weaker condition that the operator $\Gamma^{(N)}_1(\rho^1)-\Gamma^{(N)}_2(\rho^2)$ tends to 0 in trace norm.

Define the families of maps

\be
\Gamma_1^{(N)}\equiv\frac{1}{N+1}\sum_{k=0}^N(\Lambda_1\circ\Lambda_2)^k;
\Gamma_2^{(N)}\equiv\Gamma_1^{(N)}\circ\Lambda_1.
\ee

\noindent Then it can be checked that

\be
\Gamma_1^{(N)}(\Omega_1(\rho))+\frac{1}{N+1}(\Lambda_1\circ\Lambda_2)^{N+1}(\rho)=\Gamma_2^{(N)}(\Omega_2(\rho))+\frac{1}{N+1}\rho.
\ee

\noindent Given an arbitrary 1-partite behavior $q(a|x)$ for Alice we thus have that the states

\begin{eqnarray}
\sigma^1_a(N)&&\equiv\frac{N+1}{N+2}\Gamma_1^{(N)}(\rho^1_a)+q(a|1)\frac{1}{N+1}(\Lambda_1\circ\Lambda_2)^{N+1}(\rho),\nonumber\\
\sigma^2_a(N)&&\equiv\frac{N+1}{N+2}\Gamma_2^{(N)}(\rho^2_a)+q(a|2)\frac{1}{N+1}\rho
\end{eqnarray}

\noindent satisfy condition (\ref{q_causality}) and so give rise to the non-relativistic behavior

\be
P_N(a,b|x,y)=\frac{N+1}{N+2}P(a,b|x,y)+\frac{1}{N+2}q(a|x)p(b|y).
\ee

As you can see, $\lim_{N\to \infty}P_N(a,b|x,y)=P(a,b|x,y)$, i.e., $P(a,b|x,y)$ can be approximated arbitrarily well by a non-relativistic behavior.

\end{proof}

Unfortunately, the former trick does not work in scenarios other than the two inputs-two outputs case on Alice's side. Imagine, for instance, that Alice can apply two different interactions $1$ and $2$, the former having two possible outputs and the latter, \emph{three} outputs. Let $P(a,b|x,y)\in Q'$ be Alice and Bob's relativistic distribution. Then, using a similar construction one arrives at

\be
P_N(a,b|x,y)\equiv(1-p_N)P(a,b|x,y)+p_Nq(a|x)p(b|y)\in Q,
\ee

\noindent with $p_N\equiv \frac{1+\frac{3}{N+1}}{7+\frac{3}{N+1}}$. Note that this time $\lim_{N\to \infty}p_N=1/7\not=0$, i.e., we cannot prove that $P(a,b|x,y)\in \bar{Q}$.

However, this last result has a clear physical interpretation: suppose that Alice and Bob are performing an AQFT experiment of non-locality, with Alice having two interactions with two and three outcomes to play with; and Bob, an arbitrary number of interactions with an arbitrary number of outcomes. Then, if Alice only measures her system with probability $6/7-\epsilon$ and with probability $1/7+\epsilon$ outputs a random result, the resulting behavior admits a tensor representation.

\section{Conclusion}

In this article we have introduced the axiom of quansality, a generalization of the no-signalling principle where local statistics are constrained to be quantum. In the bipartite case, we showed that quansality is mathematically equivalent to the existence of a non-relativistic model to account for Alice and Bob's observations. We then showed how to exploit such an equivalence in order to recover Tsirelson's theorem and its extension to infinite dimensional systems where Alice has a two-setting/two-outcome scenario. The particular strategy that we described to `quansalize' Bob's system does not work if Alice's number of measurement settings or outcomes is bigger, though. Any further generalization of Tsirelson's theorem will then have to modify Bob's measurement operators as well as Alice's in order to move from $Q'$ to $Q$. It is also worth mentioning that the approach described here cannot be extended to the tripartite case. Indeed, reference \cite{gleason2} gives examples of tripartite distributions $P(a,b,c|x,y,z)$ which admit local quantum descriptions but such that $P(a,b,c|x,y,z)\not\in Q,Q'$.

On a more positive side, the fact that bipartite results can be derived from such simple notions suggests that the full solution of the problem may be achieved with physical intuition alone, without resorting to sophisticated mathematics. So now that they have no excuse for not coping with the current state-of-the-art on Tsirelson's problem, we would thus like to encourage physicists around the world to actually \emph{work} on it.

\end{document}